%% file: paper.tex
\documentclass[sigconf]{acmart}

\usepackage{balance}

\def\BibTeX{{\rm B\kern-.05em{\sc i\kern-.025em b}\kern-.08emT\kern-.1667em\lower.7ex\hbox{E}\kern-.125emX}}
    
\copyrightyear{2019} 
\acmYear{2019} 
\setcopyright{acmlicensed}
\acmConference[SPAA '19]{31st ACM Symposium on Parallelism in Algorithms and Architectures}{June 22--24, 2019}{Phoenix, AZ, USA}
\acmBooktitle{31st ACM Symposium on Parallelism in Algorithms and Architectures (SPAA '19), June 22--24, 2019, Phoenix, AZ, USA}

\usepackage{booktabs,epstopdf} 

\usepackage{multirow,rotating}
\usepackage{varwidth}
\usepackage{caption}
\usepackage[hang,flushmargin]{footmisc}

\usepackage{xcolor}
\usepackage{enumitem}
\usepackage{footmisc}
\usepackage{array}
\usepackage{amsmath,amsthm}

\newcolumntype{M}[1]{>{\centering\arraybackslash}m{#1}}

\def \jthRound {j}
\newcommand{\gj}{\gamma_{\jthRound}}
\newcommand{\lfrac}[2]{{#1}/{#2}}

\newcommand{\srank}[1]{\chi({#1})}
\newcommand{\skey}[1]{S({#1})}

\newif\ifapproxhistogram

\newif\ifchangasorting
\changasortingtrue

\newif\ifshowproofs
\showproofstrue

\newcommand{\compactcaption}[1]{\vspace{-1.3em}\caption{#1}\vspace{-1.35em}}

\begin{document}
\title[Histogram Sort with Sampling]{Histogram Sort with Sampling}
\author{Vipul Harsh}
\affiliation{%
  \institution{University of Illinois at Urbana Champaign}
}
\email{vharsh2@illinois.edu}

\author{Laxmikant Kale} 
 \affiliation{%
  \institution{University of Illinois at Urbana Champaign}
 }
\email{kale@illinois.edu}

\author{Edgar Solomonik} 
 \affiliation{%
  \institution{University of Illinois at Urbana Champaign}
 }
\email{solomon2@illinois.edu}

\begin{abstract}
To minimize data movement, state-of-the-art parallel sorting algorithms use techniques based on sampling and histogramming to partition keys prior to redistribution.
Sampling enables partitioning to be done using a representative subset of the keys, while histogramming enables evaluation and iterative improvement of a given partition.
We introduce Histogram sort with sampling (HSS), which combines sampling and iterative histogramming to find high-quality partitions with minimal data movement and high practical performance.
Compared to the best known (recently introduced) algorithm for finding these partitions, 
our algorithm requires a factor of $\Theta(\log(p)/\log\log(p))$ less communication, and substantially less when compared to standard variants of Sample sort and Histogram sort.
We provide a distributed-memory implementation of the proposed algorithm, compare its performance to two existing implementations, and provide a brief application study showing benefit of the new algorithm.
\end{abstract}

%

\ccsdesc[500]{Computing methodologies~Massively parallel algorithms}

\keywords{parallel sorting; data partitioning; histogramming; sampling;}  

\maketitle

\input{body}

\bibliographystyle{ACM-Reference-Format}
\bibliography{paper}

\appendix

\input{appendix}

\end{document}

%% file: body.tex

\section{Introduction}

Finding a global partition of the data is the key challenge that separates parallel sorting from sequential sorting. Partition-based sorting algorithms, that partition the data prior to redistributing it (in contrast to merge-based sorting algorithms), are advantageous on modern architectures due to their low communication cost. Sampling data either uniformly or selectively and histogramming the split produced by the partition are the two most common techniques for determining a good partition. By quantifying the parallel execution cost in terms of computation and communication, we demonstrate that a simple but careful combination of these two techniques leads to an algorithm that provides both theoretical and practical improvements over the best previously known algorithm.


A parallel sorting algorithm needs to redistribute $N$ keys across $p$ processors such that they are in a globally sorted order. In such an order, keys on processor $k$ are no greater than keys on processor $k+1$ and keys are sorted within each processor. An exact splitting (we use the terms partitioning and splitting interchangeably) is achieved if all processors own the same number of keys, while an approximate splitting guarantees that every processor owns no more than $N(1+\epsilon)/p$ keys for some $\epsilon$; we call this an $\epsilon-$balanced partition. Given sorted keys with an approximate splitting for $\epsilon=\mathcal{O}(1)$, an exact splitting can be achieved at no cost in asymptotic running time. However, it increases the running time in practice and is often not required by applications. Algorithms that guarantee a balanced partition for a given $\epsilon$ are favorable since a large $\epsilon$ increases the memory footprint and can hurt application performance.


The most natural way to cheaply determine a global partition is to collect a sample of keys, and infer a global partition from the ideal partition of the sorted sample. Sample sort~\cite{frazer1970samplesort} and its variants are basic realizations of this approach, which are widely used in practice~\cite{o2008terabyte}, and also serve as building blocks for our algorithm. 
Selecting a random sample of the data and partitioning the input key space using the random sample suffices to achieve the desired load balance w.h.p.\footnote{with high probability. In our context, $\geq (1-\mathcal{O}(p^{-c}))$ for some fixed $c>0$}  so long as $\Theta(p\log N/\epsilon^{2})$ keys are collected in the sample
~\cite{huang1983parallel}.
A deterministic balanced splitting is also possible via sampling, for example, using sample sort with regular sampling~\cite{li1993versatility,shi1992parallel}. With regular sampling, the algorithm collects $p/\epsilon$ keys from each processor, that partition the local input data on each processor evenly, requiring a total sample size of $\Theta(p^{2}/\epsilon)$ from all processors for a balanced split.
However, these classical results leave substantial room for improvement. We show that it suffices to collect a number of samples that scales near-linearly with $p$ and logarithmically with $1/\epsilon$.


Histogram sort~\cite{kale1993comparison,solomonik2010highly}, which embodies the histogramming technique, iteratively refines a partition (set of splitters), by repeatedly collecting histograms of the total number of input keys in each interval induced by the latest set of splitters. 
The number of histogramming rounds required to determine all splitters within the allowed threshold is bounded by $O(\log\mathcal{Z})$, where $\mathcal{Z}$ is the size of the input domain. For skewed distributions, the number of rounds could be large and the use of the input domain implies that Histogram sort is not a pure comparison-based sorting algorithm. 

Recently, Axtmann et al.~\cite{axtmann2015practical} proposed a scheme based on the histogram of the partition induced by a random sample.
They show that using $\Theta(p(\log (p)+1/\epsilon))$ samples results in an $\epsilon$-balanced partition w.h.p..
Our main contribution is demonstrating that by using $\log (\log (p)/\epsilon)$ steps of refinement with histogramming, $\Theta(p\log (\log (p)/\epsilon))$ samples in total suffice for an $\epsilon$-balanced partition. Our algorithm improves the communication cost for the partitioning step (proportionally to the reduction in sample size), at the cost of a small increase in the number of parallel steps (BSP supersteps / synchronizations).
This factor of improvement also holds if the partitioning schemes are used in a multi-stage fashion, for example by first splitting the data into $\sqrt{p}$ parts, then sorting each part recursively with $\sqrt{p}$ processors.

The improvement in cost warrants the introduction of an algorithm that combines sampling and histogramming, which we call Histogram sort with sampling (HSS). HSS carefully weaves together standard techniques in such a way that the resulting algorithm is provably better than the state of the art. The analysis of the algorithm is nontrivial;  
the main challenge resolved in this paper is in identifying and proving an invariant that shows global quadratic convergence of the partitioning algorithm. The main intuition behind the algorithm and proof comes from consideration of {\it splitter intervals}, which are subranges around ideal splitter keys in the globally sorted order of input keys. In each round, HSS uniformly samples keys in the union of all splitter intervals, then tightens   each splitter interval using a histogram of the collected sample. Our main analytical result is that the size of the union of the splitter intervals decreases geometrically with the number of rounds.

By characterizing HSS, we establish the theoretical soundness of iterative histogramming as a technique, that is known to be effective~\cite{kale1993Charm++,solomonik2010highly} in practice.
Our algorithm is simple, provably robust to arbitrary distributions with repeated keys, and effective in practical scenarios where the input is already partially sorted.
We provide a parallel Charm++ implementation of the HSS algorithm and demonstrate improvements over one of the fastest publicly-available distributed-memory parallel sorting algorithms, HykSort~\cite{sundar2013hyksort} in both single-staged and multi-staged settings.
Additionally, we show that our algorithm improves performance with respect to Histogram sort within the ChaNGa $N$-body code~\cite{jetley2008massively}, which uses sorting every time-step to distribute moving cosmological bodies along a space-filling curve.
Our theoretical analysis of parallel execution cost, comparative performance evaluation, and application case study unanimously identify HSS as the preferred parallel sorting algorithm. We have made our code available online~\cite{hssCode}.

\section{Problem Statement}

Let $A(0),\ldots,A(N-1)$ be an input sequence distributed across $p$ processing elements, such that each processor has $N/p$ keys. We assume that there are no duplicates in the input. In Section~\ref{duplicates} we discuss how to reduce a sorting problem with duplicate keys to a sorting problem with no duplicates, with very little overhead. 
Our proofs and algorithm also translate to scenarios where input keys are not evenly distributed across each processor. 
 Parallel sorting corresponds to redistributing and reordering the elements so that the $i$th processor owns the $i$th subsequence of 
keys in the sequence $I(0),\ldots, I(N-1)$, where $\{I(0),\ldots, I(N-1)\}=\{A(0),...,A(N-1)\}$ and $I(i)\leq I(i+1)$. 
We say that key $A(j)=I(r)$ has rank $r$.
In practice, keys are typically associated with values, but in the context of the algorithms we study, handling values of a given size is straightforward.

It is common to additionally require that the resulting distribution of sorted keys is load balanced among processors.
We compare algorithms with the standard assumption that the distribution is {\it locally balanced}, i.e. each processor owns no more than $(1+\epsilon)N/p$ keys.
However, our algorithm achieves a stronger guarantee, namely that the distribution is {\it globally balanced}, i.e. processor $i$ owns all keys greater than or equal to $\skey{i}$ and less than $\skey{i+1}$, where each $\skey{i}$ is a {\it splitter} that satisfies $\skey{i} = I(\srank{i})$, with
 \begin{align*}
 &  \quad \quad \quad \quad \quad \quad \quad \srank{i}\in\mathcal{T}_{i}, \\
 & \text{where {\it target range: }}\ \mathcal{T}_{i} = \Big[\frac{Ni}{p}{-}\frac{N\epsilon}{2p}, \frac{Ni}{p}{+}\frac{N\epsilon}{2p}\Big]
 \end{align*}
One practical advantage of a globally balanced distribution in the context of iterative applications, is that if the initial distribution is nearly sorted and globally balanced, the data exchange step is guaranteed to require little data movement.

We note that, given either type of load-balanced splitting, post-processing may be done to obtain an {\it exact} splitting~\cite{cheng2006novel}. For a locally balanced distribution, this might require some processors to potentially communicate all of their data to one or two other processors. However, given a globally balanced distribution, achieving an exact splitting would require communicating only at most $N\epsilon/p$ keys per processor.
Therefore, a more fundamental distinction is in whether a parallel sorting algorithm {\it maintains} load balance at all times, i.e. no processor is ever assigned more than $(1+\epsilon)N/p$ keys. Satisfying this condition permits bounded memory footprint, which is desirable for a parallel sorting library implementation.

The focus of this paper is on the data-partition step of partition-based sorting algorithms. Sample sort by regular sampling~\cite{shi1992parallel,li1993versatility}, histogram sort~\cite{kale1993comparison,solomonik2010highly}, sample sort by random sampling~\cite{frazer1970samplesort,blelloch1998experimental}, parallel sorting by over partitioning~\cite{li1994parallel}, AMS sort~\cite{axtmann2015practical}, HyKSort~\cite{sundar2013hyksort} fall into this category. Partition-based sorting algorithms determine a set of splitters that achieve either a locally or globally balanced splitting, then redistribute keys. The algorithm can run in multiple {\it stages} by splitting up data among subsets of processors  and sorting recursively within each subset.
In section~\ref{running-times},  we evaluate the time complexity of HSS and a multi stage variant of HSS using the standard 
  Bulk Synchronous Parallel (BSP) model by Valiant~\cite{valiant1990bridging}.%
\section{Related work}
Sample sort~\cite{frazer1970samplesort} and histogram sort~\cite{kale1993comparison} are closely related to our algorithm, we review these and other sorting algorithms before proceeding to our main result. 

\subsection{Sample sort} \label{Sample sort}
Sample sort~\cite{frazer1970samplesort,blelloch1991comparison,helman1998new,li1993versatility,shi1992parallel} is a standard well studied parallel sorting algorithm. Sample sort samples $s$ keys from each processor, and sends them to a central processor to form an overall sample of size $M = ps$ keys. Let $\Lambda = \{\lambda_{0}, \lambda_{1} ..., \lambda_{ps-1}\}$ denote the combined sorted sample. Sample sorting algorithms choose $p-1$ keys from $\Lambda$ as the final splitters. Generally, sample sort algorithms consist of the following three-phase skeletal structure.
\begin{enumerate}[leftmargin=*]
\item{\textbf{Sampling Phase}}:
Every processor samples $s$ keys and sends it to a central processor. $s$ is often referred to as the oversampling ratio. See Section \ref{Sample sort:details} for sampling methods. 

\item{\textbf{Splitter Determination}}:
 The central processor receives samples of size $s$ (from Step 1) from every processor resulting in a combined sample $\Lambda$ of size $(ps)$. The central processor then selects splitter keys: $\mathcal{S} = \{\skey{1}, \skey{2} ..., \skey{p-1}\}$ from  $\Lambda$  by picking evenly spaced keys from $\Lambda$. The splitters partition the key range into $p$ ranges, each range assigned to one processor.  Once chosen, the splitters are broadcast to all processors.

\item{\textbf{Data Exchange}}:
 Once a processor receives the splitter keys,  it sends each of its keys to their destination processor. As discussed earlier, a key in range $[\skey{i}, \skey{i+1})$ goes to processor $i$. This step is akin to one round of all-to-all communication and places all input data onto their assigned destination processors. Once a processor receives all data that is assigned to it, it merges them using a sequential algorithm, like merge sort.

\end{enumerate}
\subsection{Sample sort: Sampling methods} \label{Sample sort:details}
We discuss two sampling methods- random sampling and regular sampling, for the sampling phase (step 1) of sample sort.

\subsubsection{Random sampling}

With random sampling as described by Blelloch et al.~\cite{blelloch1998experimental}, each processor divides its local sorted input into $s$ blocks of size $(N/ps)$ and samples a random key in each block, where $s$ is the oversampling ratio. The splitters are chosen by picking evenly spaced keys from the overall sample of size $ps$, collected from all processors. Of particular reference to our work is the following theorem, (Lemma $B.4$ in~\cite{blelloch1998experimental}).
\begin{theorem}
With $\mathcal{O}\big(\frac{p\log N}{\epsilon^{2}}\big)$ samples overall, sample sort with random sampling achieves $(1+\epsilon)$ load balance w.h.p..
\end{theorem}

\subsubsection{Regular sampling}
With regular sampling~\cite{li1993versatility, shi1992parallel}, 
every processor deterministically picks $s$  evenly spaced keys from its local sorted data. The central processor collects these samples and selects splitters from this sample, just like random sampling. 
We reproduce the following theorem from ~\cite{li1993versatility,shi1992parallel}.
\begin{theorem}
If $s=\frac{p}{\epsilon}$ is the oversampling ratio, then sample sort with regular sampling achieves $(1+\epsilon)$ load balance.
\end{theorem}
Because of the large number of samples required, the sampling phase is unscalable for regular sampling. 
Sample sort with random sampling is more efficient, but scalability is still hindered in practice because of the large sample size required to achieve a load-balanced splitting.

\subsection{Histogram Sort}
Histogram sort~\cite{kale1993comparison,solomonik2010highly} addresses load imbalance by determining the splitters more reliably. Instead of determining all splitters using one large sample, it maintains a set of candidate splitter keys and performs multiple rounds of histogramming, refining the candidates in every round. 
Computing the histogram of a set of candidate keys gives the global rank of each candidate key. This information is used by the algorithm to finalize splitters or to refine the candidate splitter keys.
Once all the splitters are  within the given threshold, it finalizes the splitter keys from the set of candidate keys. The data exchange phase of Histogram sort is identical to the third phase of sample sort. 
We give an overview of the splitter determination step in histogram sort.

\begin{enumerate}[leftmargin=*]
\item  The central processor broadcasts a probe consisting of $M$ sorted keys to all processors. Usually, the initial probe is spread out evenly across the key range (unless additional distribution information is available).
\item Every processor counts the number of keys in each range defined by the probe keys, thus, computing a local histogram of size $M$.
\item Local histograms are summed to obtain a global histogram at a central processor using an $M$-item reduction.
\item The central processor finalizes and broadcasts the splitters if a probe key within the desired range has been found for each of the $p-1$ unknown splitters. Otherwise, it refines its probes using the histogram obtained and broadcasts a new set of probes for the next round of histogramming, in which case the algorithm loops back to Step 2.
\end{enumerate}


Candidate keys are refined by splitting the input key range between successive candidate keys according to their ranks~\cite{solomonik2010highly}.
Histogram sort is guaranteed to achieve any arbitrary specified level of load balance. It is also scalable in practice for many input distributions, since the size of the histogram every round is typically kept small - of the order $O(p)$. The number of histogramming rounds required to determine all splitters within the allowed threshold is at most $\log_2(\mathcal{Z})$, where $\mathcal{Z}$ is the range of the input i.e. maximum key minus the minimum key (treating $\epsilon$ as a constant here).  The number of rounds can be large, especially for skewed input distributions. 
Histogram sort has been successfully employed in real world, highly parallel scientific applications, for instance \textit{ChaNGa}~\cite{jetley2008massively}.

\subsection{Other Sorting Algorithms}

In parallel sorting by over partitioning~\cite{li1994parallel}, proposed for shared memory multiprocessors, every processor picks a random sample of size $pks$ from its local input and sends it to a central processor. The central processor sorts the overall collected sample and choses $pk-1$ splitters by selecting the $s^{th}, 2s^{th}, ..., (pk-1)s^{th}$ keys. These splitters partition the input into $pk$ buckets, more than required. 
The splitters are made available to all processors and the local input is partitioned into sublists based on the splitters. These sublists form a task queue 
and each processor picks one task at a time and processes it by copying the data to the appropriate position in the memory, determined using the splitters. 
The idea of over partitioning is closely related to histogramming. 
Recent work on sorting algorithms for asymmetric read and write costs~\cite{asymmetricSorting} and low cache complexity~\cite{cacheObliviousAlgorithms} are complimentary to our work and can be used in combination with HSS.

\subsubsection{Merge based parallel sorting algorithms}
In this paper, we primarily focus on partition-based algorithms. \textit{Merge-based} algorithms are another class of sorting algorithms that merge data in parallel using sorting networks. An early result was due to Batcher~\cite{batcher1968sorting} which uses time (or equivalently depth in a sorting network) $\mathcal{O}(\log^{2} N)$ with $N$ processors. The AKS network~\cite{ajtai19830} was the first sorting circuit of depth $\mathcal{O}(\log N)$, but had large constants because of the use of expander graphs~\cite{cole1988note,paterson1990improved}. Later, Cole~\cite{cole1988parallel} proposed a circuit that also ran in $\mathcal{O}(\log N)$ time using $N$ processors, but had smaller constants. A communication optimal algorithm in the BSP model was proposed by Goodrich~\cite{goodrich1999communication}. Cole's merge sort and its adaptation to BSP by Goodrich follow a merge-tree, but employ sampling to determine a partition that accelerates merging. Overall, unless data-partitioning schemes are also employed, merge-based algorithms tend to be less performant due to their need for more BSP supersteps for the expensive data-exchange step and in some cases more communication than partition-based alternatives.

\subsection{Large scale parallel sorting algorithms}

Several recent works have focused on large scale sorting. HykSort~\cite{sundar2013hyksort}, a state of the art practical algorithm, employs multi-staged splitting and communication to achieve better scalability. HykSort is a hybrid of sample sort and hypercube quick sort. Even though Hyksort's algorithm for splitter selection also uses sampling and histogramming, there is a key difference in the sampling method between HSS and HykSort (see Section~\ref{k rounds hss: algorithm}). As we show in Appendix~\ref{hyksort analysis}, this is critical for the running time as HykSort requires at least $\Omega \big(\log (p)\big/\log^{2}\log (p)\big)$ times more samples than HSS in the worst case. Our experiments confirm faster convergence in HSS and benefits of HSS over Hyksort in both single-staged and multi-staged settings (see section~\ref{experiments}).

AMS-sort~\cite{axtmann2015practical} employs overpartitioning for splitting. AMS-sort's scanning algorithm (Lemma 2 of~\cite{axtmann2015practical}), used to select splitters, is better than HSS with one round of histogramming by a factor of $\Theta(\min(\log p,1/\epsilon ))$. However, HSS with multiple rounds of histogramming is more efficient than AMS-sort. The scanning algorithm does not easily generalize to multiple rounds of histogramming. Further, HSS achieves a globally balanced partition, while AMS-sort achieves only a locally-balanced splitting, providing less robustness in preservation of existing distributions.
AMS-sort can be performed in a multi-stage fashion, with successive steps of splitting and data exchange across a decreasing set of processors. HSS can run in the same multi-stage fashion, but with each data partitioning step done with multiple rounds of histogramming. We compare the asymptotic running times of AMS-sort and HSS with multiple stages in Table~\ref{running times table multistage}. 



\begin{figure*}
  \begin{minipage}{0.95\textwidth}
    \centering
      \hskip 0.1 cm
       \vskip -0.2 cm
\includegraphics[width=0.9\textwidth,height=95px]{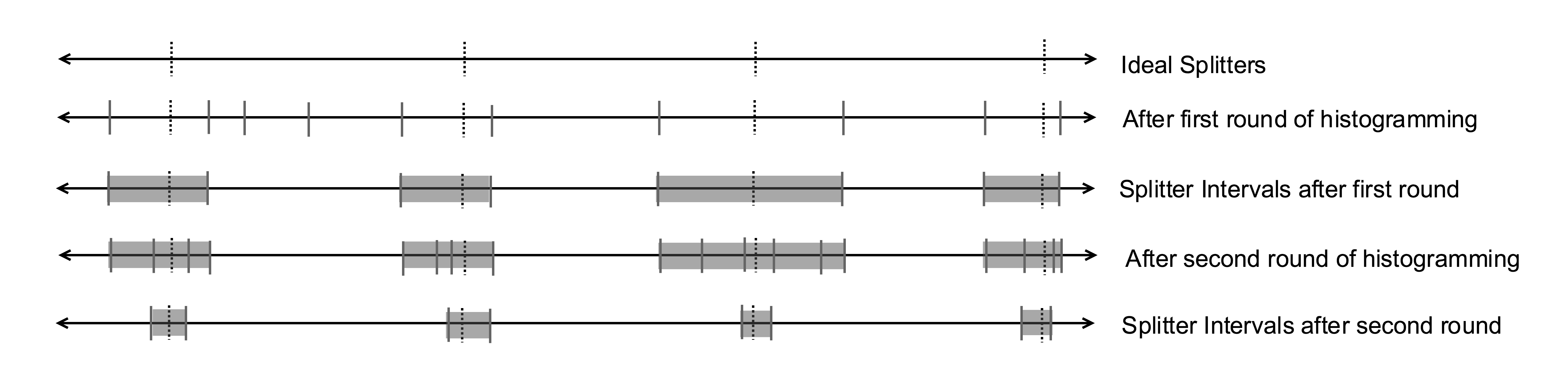} 
\compactcaption{Figure illustrating HSS with multiple rounds. After first round, samples are picked only from the splitter intervals, in proportional to the interval length. Notice how the splitter intervals shrink as the algorithm progresses.}
\label{figure:splitterIntervals}
  \end{minipage}

\end{figure*}

\subsection{Single stage AMS sort} \label{ams algorithm}






The single stage AMS-sort~\cite{axtmann2015practical} collects a single set of samples, performs one round of histogramming, then picks a locally balanced splitting based on the histogram.
The splitters obtained after the first histogramming round achieve the specified level of load balance w.h.p. with an oversampling parameter that is much less than sample sort with random sampling. In Section~\ref{running-times}, we show that the cost of histogramming is asymptotically same as the cost of sampling an equal number of keys, so AMS sort achieves a clear theoretical improvement over sample sort.
We review the AMS algorithm in some detail, due to its close relation to our approach.

\subsubsection{Scanning algorithm}
AMS sort uses a scanning technique to decide the splitters once the histogram is obtained. The algorithm scans through the histogram and assigns a maximal number of consecutive buckets (all keys between two consecutive keys in the total sorted sample) to each processor. Specifically, after assigning $i$ buckets to the first $j$ processors, it assigns buckets $i+1,\ldots, i+k$ to processor $j+1$, where $k\geq 0$ is picked maximally so that the total load on processor $j+1$ does not exceed $N(1+\epsilon)/p$. The last processor gets all the remainder elements. If the sample size is sufficiently large, the average load on the first $p-1$ processors is greater than $N/p$ w.h.p.. 

In particular, a sample of size $\Theta(p(\log p+ 1/\epsilon))$ is necessary to achieve a locally balanced partitioning w.h.p.. Demonstrating this formally is difficult due to the conditional dependence of loads assigned to consecutive processors. We formalize the proof of a key lemma in the analysis of scanning algorithm in~\cite{axtmann2015practical} (see appendix A). 
In Table~\ref{running times table} we compare the cost of AMS sort to versions of sample sort and HSS. AMS sort achieves a lower asymptotic complexity than HSS with a single round of histogramming. However, HSS can achieve an asymptotically lower complexity in $\mathcal{O}(1)$ BSP supersteps and even lower complexity with $\mathcal{O}(\log \log p/\epsilon)$ supersteps while at the same time providing a globally balanced distribution.

\section{Histogram Sort with Sampling}

The basic skeleton of HSS is similar to that of Histogram Sort. In addition, HSS employs sampling to determine the candidate probes for histogramming. Every histogramming round is preceded by a sampling phase where each processor samples local keys and the overall sample collected from all processors is used for the histogramming round. By histogramming on the sample, HSS requires significantly fewer samples compared to sample sort. 

\subsection{HSS with one histogram round} 
\label{one round}

We first describe HSS with one round of histogramming, whose data-partitioning step is slightly less efficient than AMS sort, by a factor of $\Theta(\min(\log p,1/\epsilon ))$. However, HSS achieves a globally-balanced splitting because of which HSS with one round is easily generalizable to multiple rounds of histogramming, which we discuss in subsequent sections and is the main contribution of this paper.
Extending the single round scanning algorithm of AMS sort to multiple rounds for improved complexity is non-trivial.  With multiple rounds of histogramming, HSS is more efficient than the scanning algorithm of AMS-sort, in fact, $\mathcal{O}(1)$ rounds of histogramming suffice for an asymptotic improvement.

Recall that in HSS, 
a satisfactory $i$th splitter (in terms of global load balance) is found when a candidate key is found that is known to have rank in target range $\mathcal{T}_{i}= [Ni/p - N\epsilon/2p, Ni/p + N\epsilon/2p]$. If for each target range $\mathcal{T}_{i}$, the sample contains at 
least one key with rank in $\mathcal{T}_{i}$, then after histogramming on the sample, all such splitters will be found. Intuitively, the algorithm should sample adequate number of keys so that at least one key is picked from each $\mathcal{T}_{i}$ w.h.p..




\begin{lemma} \label{one round lemma}
If every key is independently picked in the sample with probability, $\frac{ps}{N} = \frac{2p\ln p}{\epsilon N}$, where $s$, the oversampling ratio is chosen to be $\frac{2\ln p}{\epsilon}$, then at least one key is chosen with rank in $\mathcal{T}_{i}$ for each $i$, w.h.p..
\end{lemma}





\begin{proof} Recall that the input set is denoted by $A$. The size of $\mathcal{T}_{i} = N\epsilon/p$.  The probability that no key is chosen with rank in $\mathcal{T}_{i}$ in the overall sample for a given $i$ is given by,  
\begin{align*}
\Big(1 - \frac{ps}{N}\Big)^{|\mathcal{T}_{i}|} &= \Big(1 - \frac{2p\ln p}{\epsilon N}\Big)^{\frac{N\epsilon}{p}} 
\leqslant e^{-2\ln p} = \  \frac{1}{p^{2}}
\end{align*}
Since there are $p-1$ splitters, the probability that no key is chosen from some $\mathcal{T}_{i}$, is at most $(p-1)\times p^{-2} < 1/p$. 
\end{proof}
Lemma~\ref{one round lemma} leads us to the following theorem, showing global load balance of HSS with one round.
The theorem will also be useful in the analysis of multiple rounds of histogramming, each round of which effectively increases the oversampling ratio by collecting the same number of samples from a smaller subset of the complete set of keys.

\begin{theorem} \label{one round result}
With one round of histogramming and sample size $\mathcal{O}(\lfrac{p\log (p)}{\epsilon})$, HSS achieves $(1{+}\epsilon)$ load balance w.h.p..
\end{theorem}


\subsection{HSS with multiple rounds} \label{two rounds}
We show that HSS can be made more efficient by repeated rounds of sampling followed by histogramming. We build upon the key observation that after the first round of histogramming, samples for subsequent histogramming rounds can be intelligently chosen using results from previous rounds. 

\subsubsection{Sampling method}  \label{hss sampling method}

For the sampling phases, our algorithm chooses a sample from a subset $\gamma$ of the input. Initially, $\gamma$ represents the entire input. As the algorithm progresses, $\gamma$ gets smaller. HSS uses the following sampling method.

\textbf{Sampling Method}: Every key in $\gamma$ is independently chosen to be a part of the sample with probability $ps/N$, where we refer $s$ as the \textbf{\textit{sampling ratio}}.
The above sampling method simplifies the analysis, since sampling from disjoint intervals are independent. Note that the notion of sampling ratio is different from the oversampling ratio of sample sort and one round HSS since the size of the overall sample collected from all processors with the above method is $(ps|\gamma|/N)$ in expectation.


\begin{table}
\centering
\footnotesize
\sf
\begin{tabular}{@{}l|ll@{}}
\toprule
\multirow{5}{*}{\begin{turn}{90}\shortstack{{problem}}\end{turn}} 
& $N$ & number of keys to sort in total \\
& $p$ & number of processors sorting keys \\
& $A(i)$ & the $i$th input key \\
& $I(r)$ & key with rank $r$ in the overall global order\\
& $\mathcal{T}_{i}$ & $[\lfrac{Ni}{p}{-}\lfrac{N\epsilon}{p}, \lfrac{Ni}{p}{+}\lfrac{N\epsilon}{p}]$ is the target range for the $i$th splitter \\

                   \midrule
\multirow{6}{*}{\begin{turn}{90}\shortstack{algorithm}\end{turn}} 
& $s_j$ & the sampling ratio for the $j$th round, in particular \\
& & each key in $\gj$ is in the $j$th sample with probability $s_jp/N$ \\
& $L_{j}(i)$ & rank of largest sample key below rank $Ni/p$ after $j$ rounds \\ 
& $U_{j}(i)$ & rank of smallest sample key above rank $Ni/p$ after $j$ rounds \\ 
& $\mathcal{I}_{j}(i)$ $ =$ & $[I(L_{j}(i)), I(U_{j}(i))]$ is the $i$th splitter interval after $j$ rounds \\
& $\gj$ & the union of all splitter intervals after $j$ rounds \\
\bottomrule
\end{tabular}
\caption{Notation used in paper, index $j$ refers to Histogramming round, while $i$ is the processor index.}
\label{tab:notation}
\vskip -0.65 cm
\end{table}

\subsubsection{HSS with k histogramming rounds: Algorithm} \label{k rounds hss: algorithm}

\begin{enumerate}[leftmargin=*]
\item In the sampling phase before the first histogramming round, each input key is picked in the sample with probability $(ps_{1}/N)$, where $s_{1}$ is the sampling ratio for the first round. Samples  are collected at a central processor and broadcast as probes for the first histogramming round.

\item Every processor counts the number of keys in each range defined by the probe keys (the overall sample for the current round), thus, computing a local histogram. All local histograms are summed up using a global reduction and sent to the central processor. 

\item 
For each splitter $i$, the central processor maintains $L_{j}(i)$: the lower bound for the $i^{th}$ splitter rank after $j$ histogramming rounds, i.e. rank of largest key seen so far, which is ranked less than $Ni/p$. Likewise it maintains $U_{j}(i)$, rank of smallest key ranked greater than $Ni/p$. Once the histogram reduction results of the $j^{th}$ round are received, the central processor updates $L_{j}(i)$ and $U_{j}(i)$ and broadcasts the sample keys  $I(L_{j}(i))$, $I(U_{j}(i))$ bounding each splitter by the \textit{splitter interval} $\mathcal{I}_j(i) = [I(L_j(i)),I(R_j(i))]$. 

\item Once every processor is aware of the new splitter intervals, it begins its sampling phase for the $(j+1)^{th}$ round. 
Every key which falls in one of the splitter intervals is picked in the sample with probability $(ps_{j+1}/N)$, where $s_{j}$ denotes the sampling ratio for the $j^{th}$ round. If $j < k$, samples from all processors are collected at a central processor and broadcast for the next round of histogramming, in which case the algorithm loops back to step 2. If $j = k$, the histogramming phase is complete and the algorithm continues to step 5. Step 2, 3 and 4 can be executed efficiently if the local data is already sorted.

\item Once the histogramming phase finishes, the key ranked closest to $Ni/p$ among the keys seen so far is set as the $i^{th}$ splitter. Later, we discuss how to choose $k$ and the sampling ratios $s_{j}$'s so that the splitters determined this way result in a globally balanced partition.
\end{enumerate}

The critical difference between HykSort and HSS is in the sampling method. HykSort samples equally from all splitter intervals whereas HSS samples in proportion to the interval length. By sampling more from larger intervals, HSS is able to narrow down the intervals quicker. We show that HykSort requires at least $\Omega \big(\log (p)\big/\log^{2}\log (p)\big)$ times more samples than HSS in the worst case (see Appendix~\ref{hyksort analysis}). 

A crucial observation is that the splitter intervals shrink as the algorithm progresses and hence the sampling step is executed with a subset of the input that gets smaller every round. Let ${\gj}$ denote the set of keys in the input that belong to one of the splitter intervals after $\jthRound$ rounds. $|{\gj}|$ represents the size of the input that the algorithm samples from, for the $\jthRound^{th}$ round.
We have, $|{\gj}| \leqslant \sum_{i}U_j(i)- L_j(i)$, where $U_j(i)-L_j(i)$ is the number of keys in the $i$th splitter interval. Some splitter intervals can overlap, hence the inequality
. In fact, it is easy to see that there is no partial overlap between two splitter intervals, that is, either two splitter intervals: $\mathcal{I}_{\jthRound}(i_{1})$ and $\mathcal{I}_{\jthRound}(i_{2})$ are disjoint or they are identical.

Our proof outline is as follows.
First we show in Lemma~\ref{lemma_k_ROUNDS_stop} that the algorithm will achieve a good splitting w.h.p. if the sampling ratio for the final round (the $k^{th}$ round) is chosen to be large enough.  
Having shown that the algorithm terminates after $k$ rounds, achieving a globally load balanced partition, we bound the sample sizes in each round by first bounding $|\gj|$ in terms of the sampling ratio $s_{j}$ necessary to obtain all splitters w.h.p..
Finally, we appropriately set the sampling ratios such that the size of the union of splitter intervals, that is, $|\gj|$ decreases by a constant factor. Intuitively, sampling ratio $s_{j}$ in round $j$ (where samples are chosen only from the splitter intervals in round $j$) can be thought of as choosing samples from the entire input range with an oversampling ratio of $s_{j}$ and discarding unnecessary samples. 

%





\begin{lemma} \label{lemma_k_ROUNDS_stop}
If $s_{k} = \frac{2\ln p}{\epsilon}$ be the sampling ratio for the $k^{th}$ round, then at least one key is chosen from each $\mathcal{T}_{i}$ after $k$ rounds w.h.p..
\end{lemma}

Given a sampling ratio of $s_j=\frac{2p\ln p}{\epsilon N}$, all splitters are found w.h.p, by Lemma~\ref{one round lemma}.
We next bound the expectation of the size of the union of all splitter intervals. 
\begin{lemma} \label{lemma_k_ROUNDS_sample_length} 
Let $s_{\jthRound}$ be the sampling ratio for the $\jthRound^{th}$ round, $\mathcal{I}_{\jthRound}(i)$ be the splitter interval for the $i^{th}$ splitter after $\jthRound$ rounds  and ${\gj}$ denote the set of input keys that lie in one of the $\mathcal{I}_{\jthRound}$'s, then, 
$E(|{\gj}|) \leqslant \frac{2N}{s_{\jthRound}}
$.
\end{lemma}
\begin{proof} Since $L_{\jthRound}(i)$ and $U_{\jthRound}(i)$ are only improved every round,
 \[L_{\jthRound-1}(i) \leqslant L_{\jthRound}(i) \leqslant \frac{Ni}{p} \leqslant U_{\jthRound}(i) \leqslant U_{\jthRound-1}(i)\]
Further  
$\forall x: 0\leqslant x \leqslant \big( U_{\jthRound-1}(i)-\frac{Ni}{p}\big)$,
\begin{align*}
& P\Big[U_{\jthRound}(i) - \frac{Ni}{p} \geqslant x\Big] = \big(1 - \frac{ps_{\jthRound}}{N}\big)^{x}
\end{align*}
As a result, we can bound the size of the $i$th splitter interval,
\begin{align*}
E\Big[U_{\jthRound}(i) -  \frac{Ni}{p}\Big] 
=& \sum_{x=1}^{U_{\jthRound-1}(i)-\frac{Ni}{p}} P\Big[U_{\jthRound}(i) - \frac{Ni}{p} \geqslant x\Big] \quad \quad \quad \quad\quad \quad \\
=& \sum_{x=1}^{U_{\jthRound-1}(i)-\frac{Ni}{p}}\big(1 - \frac{ps_{\jthRound}}{N}\big)^{x} \\
\leqslant & \sum_{x=0}^{\infty} \big(1 - \frac{ps_{\jthRound}}{N}\big)^{x}  
= \frac{N}{ps_{\jthRound}}
\end{align*}
By a similar argument we have that,
$E\Big[\frac{Ni}{p} - L_{\jthRound}(i)\Big] \leqslant \frac{N}{ps_{\jthRound}}$, 
\begin{align*}
 \text{Thus,} \quad E[|{\gj}|] & \leqslant \ E\Big[\sum_{i=1}^{p-1}|\mathcal{I}_{\jthRound}(i)\cap A|\Big] = \sum_{i} E\Big[U_{\jthRound}(i)-L_{\jthRound}(i)\Big] \\
&=  \sum_{i} E\Big[\frac{Ni}{p} - L_{\jthRound}(i)\Big] + E\Big[U_{\jthRound}(i) - \frac{Ni}{p}\Big] \\
&\leqslant \sum_{i} \frac{2N}{ps_{\jthRound}} = \frac{2N}{s_{\jthRound}}
\end{align*} 
\vskip -0.64 cm
\end{proof}
%
\noindent Lemma~\ref{lemma_k_ROUNDS_sample_length} suggests that ${\gj}$ will be small in expectation. The next lemma shows that it is also small w.h.p..
\begin{lemma}\label{lemma_K_ROUNDS_BOUND}
If  $s_{\jthRound}<\sqrt{\frac{2p}{\ln p}}$, then, $|{\gj}| 
\leqslant \frac{4N}{s_{\jthRound}} \ \text{w.h.p.} 
$
\end{lemma}
\ifshowproofs
\begin{proof} 
\label{THEOREM_K_ROUNDS_BOUND_PROOF}
The main challenge in proving the above theorem is in handling the dependency in splitter intervals, for e.g. when they overlap.
We first modify the definition of splitter intervals in the following way, so that the union of the splitter intervals remains unchanged. 
\begin{align*}
U'_{\jthRound}(i) = \min\Big(\frac{N(i+1)}{p}, U_{\jthRound}(i)\Big),\ L'_{\jthRound}(i) = \max\Big(\frac{N(i-1)}{p}, L_{\jthRound}(i)\Big)
\end{align*}
The above definition effectively strips the splitter interval 
\\ $[I(L_{j}(i)), I(U_{j}(i))]$ to $[I(L'_{j}(i)), I(U'_{j}(i))]$. 
To see that stripping doesn't change the union of all splitter intervals, consider a $U_{\jthRound}(i)$ which is greater than $N(i+1)/p$. Then by definition, we have $U_{\jthRound}(i) = U_{\jthRound}(i+1)$. Thus, the portion of $\mathcal{I}_{j}(i)$ that extends beyond $N(i+1)/p$ is included in $\mathcal{I}_{j}(i+1)$. Hence, restricting $U_{\jthRound}(i)$ to $Ni/p + N/p$ doesn't change $\gamma_{j}$ - the union of $\mathcal{I}_{j}$'s. An inductive argument (by considering splitter intervals from left to right) shows that restricting all $U_{\jthRound}$'s doesn't change $\gamma_{j}$. A similar argument can be used for $L_{\jthRound}$'s.

Observe that, $U'_{\jthRound}(i)$'s are independent random variables. This is because the possible values of $U'_{\jthRound}(i_{1})$ and $U'_{\jthRound}(i_{2})$ for $i_{1} \neq  i_{2}$ are completely disjoint. 
The value of $U'_{\jthRound}(i)$ is determined completely by sampling in the interval $[Ni/p, N(i+1)/p)$ and since sampling in disjoint intervals are independent, $U'_{\jthRound}(i)$'s are independent. 


\noindent We have, \(E[U'_{\jthRound}(i)-Ni/p] \leq E[U_{\jthRound}(i)-Ni/p] \leq N/ps_{j}\). 
\begin{align*}
\text{Thus, } P  \Big[\sum_{i}  & U'_{\jthRound}(i) - \frac{Ni}{p}   > \frac{2N}{s_{j}}\Big]  \quad \quad \quad \quad  \quad \quad \\
&= P\Big[\sum_{i} U'_{\jthRound}(i) - \frac{Ni}{p} - \frac{N}{ps_{j}} > \frac{N}{s_{j}}\Big] \\
&\leq P\Big[\sum_{i} U'_{\jthRound}(i) - \frac{Ni}{p} - E\big[U'_{\jthRound}(i) - \frac{Ni}{p}\big] > \frac{N}{s_{j}}\Big] \\
& \leq e^{-\frac{2N^{2}}{s_{j}^{2}}\big/\sum_{i} (N/p)^{2}} = e^{-\frac{2p}{s_{i}^{2}}} \leq e^{-2\ln p} \\ 
&\leq \frac{1}{p^{2}} \text{  (using Hoeffding's inequality)}
\end{align*}

Note that $U'_{\jthRound}(i)-Ni/p$ lies strictly in the interval $[0,N/p]$, this fact is used in the application of the Hoeffding's inequality. On similar lines we have,
$\sum_{i} (Ni/p - L'_{\jthRound}(i)) \leqslant \frac{2N}{s_{\jthRound}},\ w.h.p. $.
We then conclude,
\begin{align*}
|{\gj}| &\leqslant  \sum_{i} \Big(U'_{\jthRound}(i) - \frac{Ni}{p}\Big) + \Big(\frac{Ni}{p} - L'_{\jthRound}(i)\Big) \leqslant \frac{4N}{s_{\jthRound}} \quad  w.h.p.
\end{align*}
\vskip -0.7 cm
\end{proof}
\else
We omit the proof for Lemma~\ref{lemma_K_ROUNDS_BOUND} for brevity.
\fi
The next lemma bounds the sample size for each round in terms of the sampling ratios.

\begin{lemma} \label{lemma_k_rounds_SampleSize}
Let $Z_{\jthRound}$ be the sample size for the $\jthRound^{th}$ round and $s_{\jthRound} \geqslant s_{\jthRound-1}$, then $Z_{\jthRound} \leqslant (5ps_{\jthRound}/s_{\jthRound-1})$ w.h.p.
\end{lemma}
\begin{proof} We have, 
$
E[Z_{\jthRound}] = |\gamma_{\jthRound-1}| ps_{\jthRound}/N
$. We also have,
$|\gamma_{\jthRound-1}| \leqslant 4N/s_{\jthRound-1} \ w.h.p.$, using Lemma~\ref{lemma_K_ROUNDS_BOUND}.

Given that $|\gamma_{\jthRound-1}| \leqslant 4N/s_{\jthRound-1}$, using Chernoff bounds, 
\begin{align*}
 P[Z_{\jthRound} \geqslant (5ps_{\jthRound}/s_{\jthRound-1})] &\leqslant
P[Z_{\jthRound} \geqslant E[Z_{\jthRound}] + ps_{\jthRound}/s_{\jthRound-1}] \\
&\leqslant e^{-\frac{(ps_{\jthRound}/s_{\jthRound-1})^{2}}{3E[Z_{\jthRound}]}} 
= e^{-\frac{(ps_{\jthRound}/s_{\jthRound-1})^{2}N}{3|\gamma_{\jthRound-1}| ps_{\jthRound}}} \\
&\leqslant e^{-\frac{(ps_{\jthRound}/s_{\jthRound-1})^{2}Ns_{\jthRound-1}}{12N ps_{\jthRound}}}
\leqslant e^{-\frac{p}{12}}
\end{align*} 
%
\vskip -0.62 cm
\end{proof}

With Lemmas \ref{lemma_k_rounds_SampleSize} and \ref{lemma_k_ROUNDS_stop} in hand, we are now prepared to discuss how to appropriately choose the sampling ratios so that our algorithm achieves the desired load balance. \\

For HSS with $k$ rounds, if we set the sampling ratio for the $j^{th}$ round as $s_{j} = (2\ln p/\epsilon)^{j/k}$, then after $k$ rounds all splitters are found w.h.p., using Lemma~\ref{lemma_k_ROUNDS_stop}. The size of the union of splitter intervals, that is, $|\gj|$ is less than $4N/s_{j} = 4N (\epsilon/2\ln p)^{1/k}$ using Lemma~\ref{lemma_K_ROUNDS_BOUND}. The sample size for the $j^{th}$ histogramming round is at most $5ps_{j}/s_{j-1} = 5p(2\ln p/\epsilon)^{1/k}$ using Lemma~\ref{lemma_k_rounds_SampleSize}. This gives us our main theorem.

\begin{table*}
\centering
\vskip -0.22 cm
\resizebox{0.76\textwidth}{!}
{%
\begin{tabular}
{|M{3.61 cm}|M{2.3 cm}|M{3.2 cm}|M{2.3 cm}|M{1.5 cm}|} \hline
Algorithm & Overall sample size  & Computation complexity & Communication complexity & Supersteps\\ \hline

Regular sampling  & $\mathcal{O}(\frac{p^{2}}{\epsilon})$ & $\mathcal{O}\Big(\frac{p^{2}}{\epsilon}\log p\log p\Big)$  & $\mathcal{O}\Big(\frac{p^{2}}{\epsilon} \Big)$ & $O(1)$ \\ \hline

Random sampling  & $\mathcal{O}(\frac{p\log N}{\epsilon^{2}})$  &  $\mathcal{O}\Big(\frac{p\log N\log p}{\epsilon^{2}}\Big)$ & $\mathcal{O}\Big(\frac{p\log N}{\epsilon^{2}} \Big)$ & $O(1)$\\ \hline

Single stage AMS sort &  $\mathcal{O}(p(\log p+ \frac{1}{\epsilon}))$   &  $\mathcal{O}\Big(p(\log p+ \frac{1}{\epsilon})\log N\Big)$ & $\mathcal{O}\Big(p(\log p+\frac{1}{\epsilon}) \Big)$ & $O(1)$\\ \hline


HSS with one round &  $\mathcal{O}(\frac{p\log p}{\epsilon})$   &   $\mathcal{O}\Big( \frac{p\log p}{\epsilon}\log N\Big)$ & $\mathcal{O}\Big(\frac{p\log p}{\epsilon}  \Big)$ & $O(1)$\\ \hline

HSS with two rounds &  $\mathcal{O}(p\sqrt{\frac{\log p}{\epsilon}})$  &   $\mathcal{O}\Big( p\sqrt{\frac{\log p}{\epsilon}}\log N\Big)$ & $\mathcal{O}\Big(p\sqrt{\frac{\log p}{\epsilon}} \Big)$ & $O(1)$\\ \hline

HSS with $k$ rounds & $\mathcal{O}(kp\sqrt[k]{\frac{\log p}{\epsilon}})$  & $\mathcal{O}\Big( kp\sqrt[k]{\frac{\log p}{\epsilon}}\log N\Big)$ & $\mathcal{O}\Big(kp\sqrt[k]{\frac{\log p}{\epsilon}}  \Big)$ & $O(k)$\\ \hline

HSS with $O(1)$ samples per processor per round 
& $\mathcal{O}(p\log \frac{\log p}{\epsilon})$ &  $\mathcal{O}\Big( p\log \frac{\log p}{\epsilon}\log N \Big)$ & $\mathcal{O}\Big(p\log \frac{\log p}{\epsilon}\Big)$ & $O(\log \frac{\log p}{\epsilon})$\\ \hline
\end{tabular}}
\caption{Cost complexity of data partitioning step of Sample sort, AMS sort, and HSS. Data exchange costs are excluded.}
\label{running times table}
\end{table*} 
\begin{table*}
\centering
\vskip -0.7 cm
\resizebox{0.76\textwidth}{!}{%
\begin{tabular}
{|M{2.61 cm}|M{2.1 cm}|M{4.8 cm}|M{4.0 cm}|M{1.7 cm}|} \hline
Algorithm & Sample size per stage  & Computation complexity & Communication complexity & Supersteps\\ \hline
AMS sort, $l$ stages & $\mathcal{O}(r(\log r + \frac{1}{\epsilon}))$ & $\mathcal{O}\Big(\frac{N}{p}\log N + lr(\log r + \frac{1}{\epsilon})\log N \Big)$ & $\mathcal{O}\Big( lr(\log r + \frac{1}{\epsilon}) + \frac{lN}{p}\Big)$ & $O(l)$ \\ \hline

HSS, $l$ stages $\mathcal{O}(\log \frac{\log r}{\epsilon})$ rounds per stage & $\mathcal{O}(r\log \frac{\log r}{\epsilon})$  & $\mathcal{O}\Big(\frac{N}{p}\log N + lr \log (\frac{\log r}{\epsilon})\log N\Big)$ & $\mathcal{O}\Big(lr\log {\frac{\log(r)}\epsilon} + \frac{lN}{p}\Big)$ & $O(l\log{\frac{\log{p}}{\epsilon}})$ \\ \hline
\end{tabular}}

\caption{Cost complexity of $l$-stage HSS and AMS sort; the size of processor group decreases by a factor of $r=p^{1/l}$ each round.}
\label{running times table multistage}
\vskip -0.7 cm
\end{table*}

\begin{theorem} \label{k rounds result}
With $k$ rounds of histogramming and a sample size of $\mathcal{O}\Big(p\sqrt[k]{\frac{\log p}{\epsilon}}\Big)$ per round  , HSS achieves $(1+\epsilon)$ load balance w.h.p. for large enough $p$\footnote{Specifically, so long as $s_j = \frac{2\ln p}{\epsilon} \leq \sqrt{\frac{2p}{\ln p}}$ for Lemma~\ref{lemma_K_ROUNDS_BOUND}}.
\end{theorem}
Observe from theorem~\ref{k rounds result} that there is a trade off between the sample size per round ($\mathcal{O}(p\sqrt[k]{\log p/\epsilon})$) of histogramming and the number of histogramming rounds. To minimize the number of samples across all rounds, we take derivative of $(kp\sqrt[k]{\log p/\epsilon})$ w.r.t. $k$ and set it to $0$,

\begin{align*}
& \frac{d(kp\sqrt[k]{\log p/\epsilon})}{dk}  = p\sqrt[k]{\log p/\epsilon}\Big(1  - \frac{\log \frac{\log p}{\epsilon}}{k}\Big) = 0 \\
\Rightarrow \quad & k = \log \frac{\log p}{\epsilon}
\end{align*}
The overall sample size $\mathcal{O}(kp\sqrt[k]{\log p/\epsilon})$ attains global minimum for $k=\log (\log p/\epsilon)$ histogramming rounds and $|\gj| \leq 4N / (e)^{j}$ at the minima using Lemma~\ref{lemma_K_ROUNDS_BOUND}. 
Across all rounds, the overall sample size from all processors is $\mathcal{O}(p\log (\log p/\epsilon))$.  This leads us to the following main theorem.

\begin{theorem} \label{optimal rounds result}
With $k=\mathcal{O}(\log (\log p/\epsilon))$ rounds of histogramming and $ \mathcal{O}(p)$ samples per round ($\mathcal{O}(1)$ from each processor), HSS achieves $(1+\epsilon)$ load balance w.h.p., for large enough $p$.
\end{theorem}

Setting $\epsilon = p/N$ results in exact splitting and hence we get the following result for exact splitting.

\begin{theorem}
HSS with $\mathcal{O}(p)$ samples per round overall can achieve exact splitting in $O(\log N/p + \log \log p)$ rounds.
\end{theorem}





\section{Running times} \label{running-times}

We model an algorithm's parallel execution as a sequence of BSP supersteps, during each of which, processors perform computation locally, then send and receive messages.
An algorithm's BSP complexity consists of three components:
\begin{enumerate}[leftmargin=*]
\item the number of supersteps (synchronization cost),
\item the sum over all supersteps of the maximum amount of computation done by any processor during the superstep (computation cost),
\item the sum over all supersteps of the maximum amount of data sent or received by any processor during the superstep (communication cost).
\end{enumerate}
We permit processors to send and receive $p$ messages every superstep, which simplifies the analysis of the all-to-all data exchanges.
The model captures the performance trade-offs for our purpose, more histogramming rounds increase number of supersteps but lower communication cost.


We analyse the computation and communication cost for sample sort and HSS. 
Both algorithms have the same cost for initial local sorting, broadcasting splitters and data exchange. The computation cost of local sorting is $\mathcal{O}((N/p) \log \frac{N}{p})$. No communication is involved in local sorting. The cost of broadcasting splitters once they are determined is $\mathcal{O}(p)$.
The final data movement requires all data to be sent to their destination processors, hence the communication cost involved is $\mathcal{O}(N/p)$. Once a processor receives all data pieces, it merges them, which takes $\mathcal{O}((N/p)\log p)$ computation time.

\subsection{Cost of Sampling}
Collecting a sample of overall size $S$ onto one processor, requires a single BSP superstep with a communication cost of $\mathcal{O}(S)$.
In practice, random sampling is usually performed with each processor selecting $S/p$ elements, and a {\it gather} collective communication protocol, which collects all samples onto one processor.
Sorting the overall sample on a central processor costs $\mathcal{O}(S\log p)$ work locally if each processor provides a sorted contribution to the sample. 

\subsection{Cost of Histogramming}

A local histogram can be computed in $\mathcal{O}(S\log (N/p))$ time using $S$ binary searches on the local sorted input, where $S$ denotes the size of the histogram. A global histogram is computed by reducing all local histograms. An $S$-item reduction requires 2 BSP supersteps (one for a reduce-scatter and one for an all-gather) with $\mathcal{O}(S)$ communication and computation~\cite{pjevsivac2007performance,thakur2003improving}.
The histogram probes and the splitter intervals are broadcast to every processor for histogramming. The communication cost of broadcasting a length $S$ message is $\mathcal{O}(S)$.
Thus, both the computation and communication costs of histogramming are proportional to the overall sample size. 

The BSP complexity of the data partitioning step of sample sort, AMS sort, and HSS are shown in Table~\ref{running times table}. 
AMS sort and HSS require significantly fewer samples due to histogramming.
We observe that the best HSS configuration has strictly superior complexity to all other considered algorithms.


\subsection{HSS with Multiple Stages}
Like AMS-sort, HSS can be generalized to a multi-stage algorithm. We refer readers to~\cite{axtmann2015practical} for details on multi-stage AMS sort. We simply consider the benefit of replacing the data-partitioning step in multi-stage AMS-sort with multiple histogramming rounds of HSS. The rest of the algorithm involving the data exchange steps is unchanged. The running time complexities of $l$ stage AMS-sort and $l$ stage HSS in the BSP model are shown in Table~\ref{running times table multistage}. In each step, a processor group gets divided into $r=p^{1/l}$ processor groups.

The first local sorting takes $\mathcal{O}(N\log (N/p)/p)$ time. At the end of each stage, every processor receives $\mathcal{O}(r)$ data pieces that it needs to merge. So, the total computation cost of local sorting after every stage excluding the first local sorting is $\mathcal{O}((lN/p)\log r) = \mathcal{O}((N/p)\log p)$. This gives an overall computation cost of local sorting as $\mathcal{O}((N\log (N/p))/p + (N\log p)/p) = \mathcal{O}((N\log N)/p)$.
The computation cost of sampling and histogramming for HSS is $\mathcal{O}(r \log ((\log r)/\epsilon))\log N)$ per stage.
Each sampling, histogramming and data exchange step takes $\mathcal{O}(1)$ BSP supersteps.
The number and cost of all of these steps is uniform throughout stages, so all of these costs are multiplied by a factor of $l$, the number of stages.
Consequently, we observe a trade-off between the cost of data partitioning and the cost of data exchanges that depends on $l$.

\section{Implementation}  \label{experiments}

\begin{figure*}
    \begin{minipage}{0.49\textwidth}
    \centering
    \vskip -0.4 cm 
    	\includegraphics[width=0.75\textwidth,height=144pt]{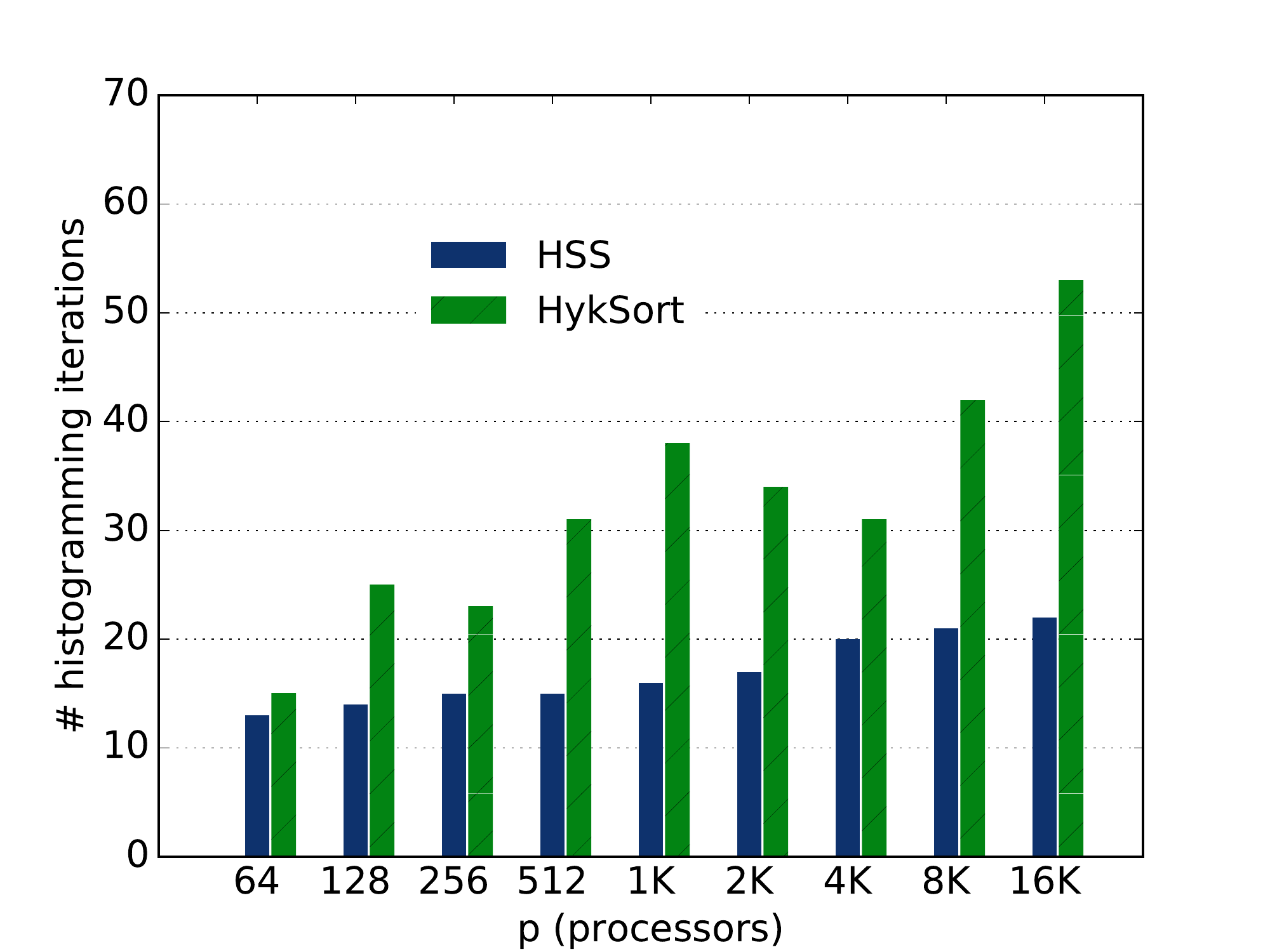}
		\compactcaption{Number of histogramming iterations for HSS vs HykSort. We used 1 sample per processor per round for this experiment. The worst case number of iterations increases more gradually in HSS ($\mathcal{O}(\log (\log p/\epsilon))$) than Hyksort ($\mathcal{\Omega}(\log (p/\epsilon)/\log \log p)$).  Experiments were run on the Stampede 2 supercomputer.} 
		\label{fig:hyk_hss_iterations}
  \end{minipage}
  \hskip 0.2 cm
  \begin{minipage}{0.49\textwidth}
    \centering
  \vskip -0.4 cm
\includegraphics[width=0.75\textwidth]{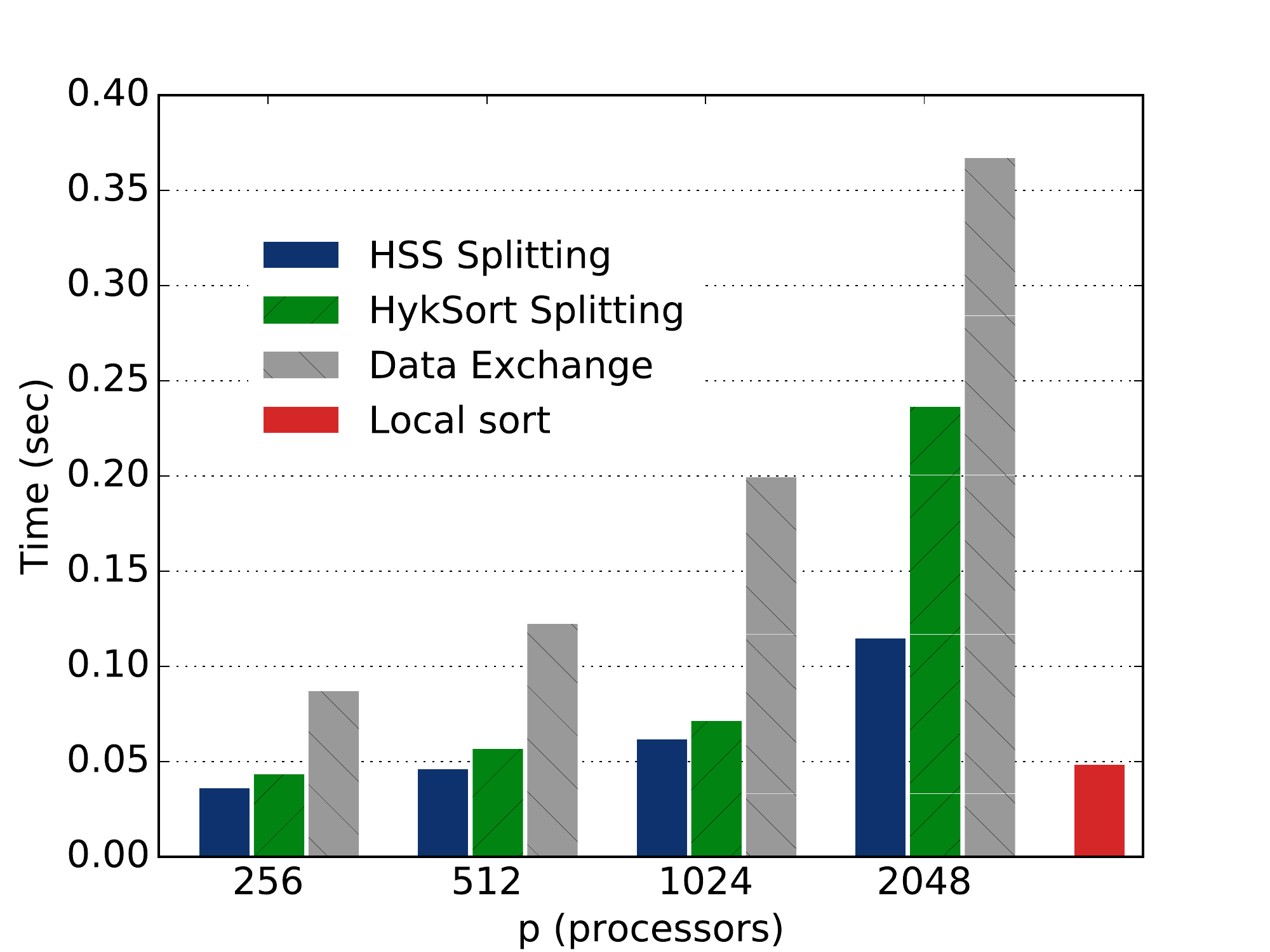}
		\compactcaption{Single staged runs with 16 threads per processor, 1M 8-byte long keys per processor.  Experiments were run on the Stampede 2 supercomputer.} 
		\label{fig:singlestage}
  \end{minipage}
\end{figure*}



We implemented HSS in C++11 in the Charm++~\cite{acun2014parallel,kale1993Charm++,bhatele20145} framework. 
Charm++ allows an application to create any number of virtual processors, called chares which are scheduled by the runtime system. Additionally, chares can be tied to a specific core or a node (called \textit{group} and \textit{nodegroup} chares in Charm++ terminoology).

 Our implementation comprises of three  phases; local sorting of input data, splitter determination using histogramming, and final data exchange. We use C++ \texttt{sort} for local sorting for the first phase. Let $t$ denote the total number of threads and $p$ denote the number of processors (processes or ranks).

\begin{itemize}[leftmargin=0.5em]
\item{\textbf{Histogramming Phase}}:
The histogramming phase determines $p-1$ splitters for process level splitting. For the sampling phase before every histogramming round, each thread samples probe keys, from its input, which lie in the union of splitter intervals. If $\delta$ denotes the fraction of input covered by the splitter intervals, then every thread picks $s/\delta$ samples from its entire input and discards samples that don't lie in any of the splitter intervals. This way the expected size of the overall sample is $s \times t$, where $s$ can be thought of as the oversampling ratio. The overall sample is assembled at the central processor and broadcast for histogramming. Every thread computes a local histogram using binary searches on it's locally sorted input. The local histograms are summed up using a reduction and sent to the central processor. 

\item{\textbf{Data Exchange}}:
Once every processor receives the splitters, input data from all threads within a processor are merged and partitioned into $p$ messages, one for each processor. 

\item{\textbf{Local Merging}}:
Once a process receives all data that falls in its bucket, it merges and redistributes data among its threads using HSS with one round.
\end{itemize}


\subsection{Handling duplicates via implicit tagging} \label{duplicates}
We use the following standard technique to deal with duplicates in the inputs. 
We enforce a strict ordering on keys by implicitly replacing each key $k$ with a triplet $(k, processor, ind)$, where $processor$ denotes the processor that $k$ resides on and $ind$ denotes the index in the local data structure, where the key is stored. 

\subsection{Multi-staged sorting and comparisons with other algorithms}
To compare HSS to Hyksort in both single-staged and multi-staged settings, we implemented the splitting algorithm of HSS in the HykSort code~\cite{hyksortCode}, written in the MPI framework. This allows a fair comparison without the side-effects of two different parallel programming frameworks, namely Charm++ and MPI.

\section{Experimental Results}
In this section, we describe our experimental results. The goal of our experiments is to demonstrate the fast splitter determination of HSS compared to other state of the art algorithms and to demonstrate its benefits in both single-staged and multi-staged settings. We also include a brief application study to supplement our results.  

\subsection{Fast convergence of splitters}
HSS determines all splitters in $\mathcal{O}(\log (\log p)/\epsilon)$ rounds using $\mathcal{O}(1)$ samples per round per processor. This results in faster convergence compared to HykSort, which requires $\Omega(\log (p/\epsilon)/\log \log p)$ rounds with the same number of keys. We ran the splitting algorithm of HSS and HykSort with 1 sample per processor per round and $\epsilon = 2\%$ to verify the same. As illustrated in Figure~\ref{fig:hyk_hss_iterations}, the number of iterations
in HSS increases gradually compared to HykSort. Note that the execution time of the splitting phase is directly proportional to the number of iterations.

\subsection{Weak scaling and comparison to HykSort} \label{single stage}
In this section we describe single-staged experiments and comparisons to HykSort. We implemented HSS's splitting algorithm in HykSort's code for the fairest comparison. For this set of experiments we used 1 million keys per processor and 16 threads per processor. We used a probe size equal to 5p per histogramming round for both HSS and Hyksort which we found to be a reasonable sample size. We also found the default sample size of Hyksort, set as per ~\cite{sundar2013hyksort} to be suboptimal for this set of experiments. 
Figure~\ref{fig:singlestage} illustrates weak scaling experiments. Besides the splitting time for the splitter determination step, the local sorting time and the data exchange times are also shown (which are common to both Hyksort and HSS). As can be observed from the figure, the difference between HSS and Hyksort's splitting phase becomes more apparent with increasing number of processors. The improved splitting of HSS results in a modest improvement of 10-15\% in the overall running time for higher number of processors.  

Single-staged AMS sort requires about $2p\ max(1/\epsilon, \log p)) \approx 100 p$ for $p=2048$ samples to achieve the desired splitting. In contrast, HSS took 6 iterations to converge with $5p$ samples per iteration resulting in about $30p$ samples overall. The execution time of the splitting phase is directly proportional to the number of samples, hence one can expect single-staged AMS to take approximately $3x$ time for the splitting phase. 

\begin{figure*}[ht]
  \begin{minipage}[b]{0.49\linewidth}
    \centering
     \vskip -0.4 cm
     	\includegraphics[width=0.7\textwidth]{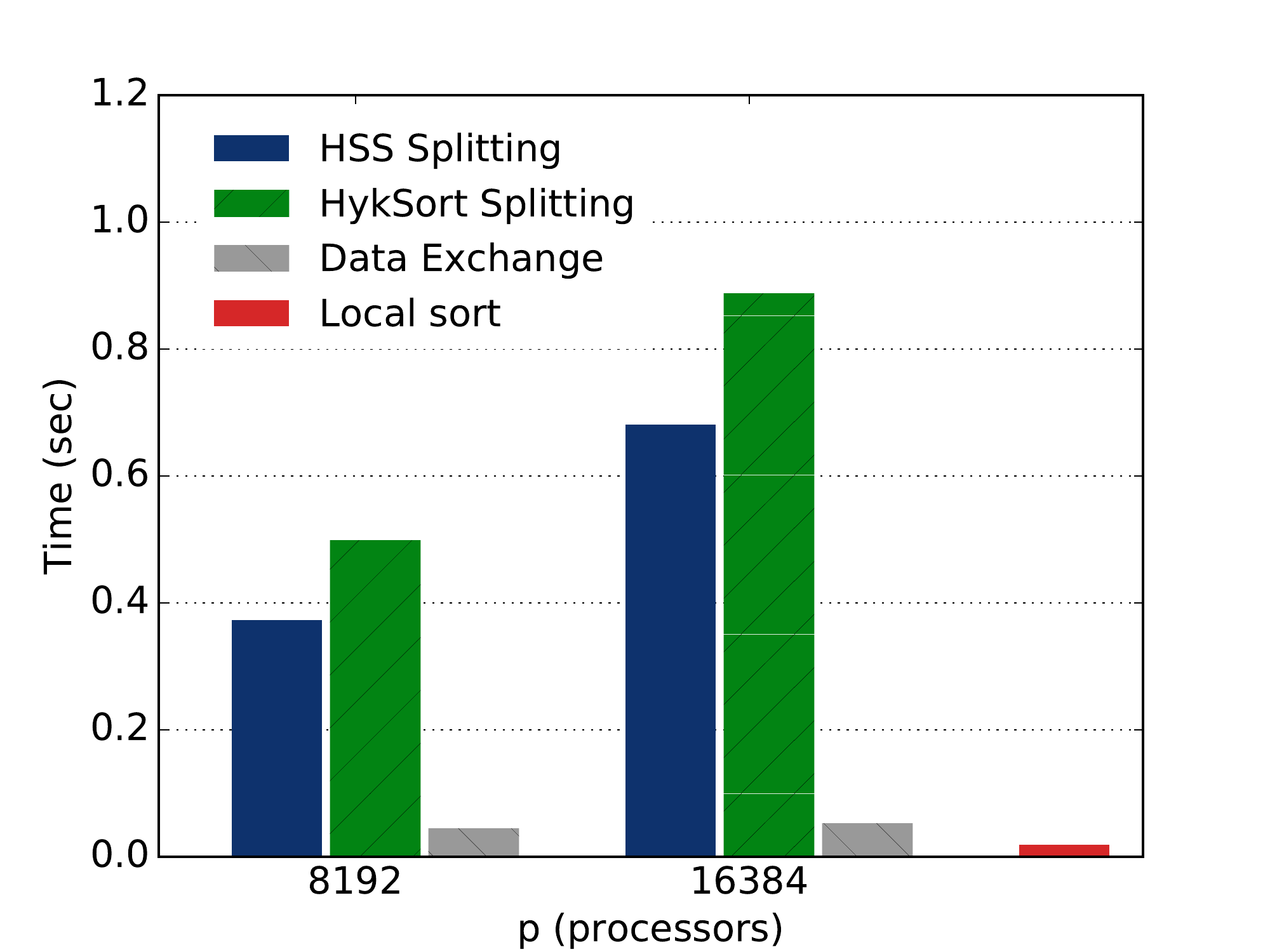}
	\compactcaption{2-staged runs with 1 thread per processor, $10^5$ 8-byte long keys per processor. Experiments were run on the Stampede 2 supercomputer.} 
	\label{fig:multiplestage}
  \end{minipage}
  \hskip 0.2cm
    \begin{minipage}[b]{0.49\linewidth}
    \centering
      \vskip -0.3cm
	\includegraphics[width=0.7\textwidth]{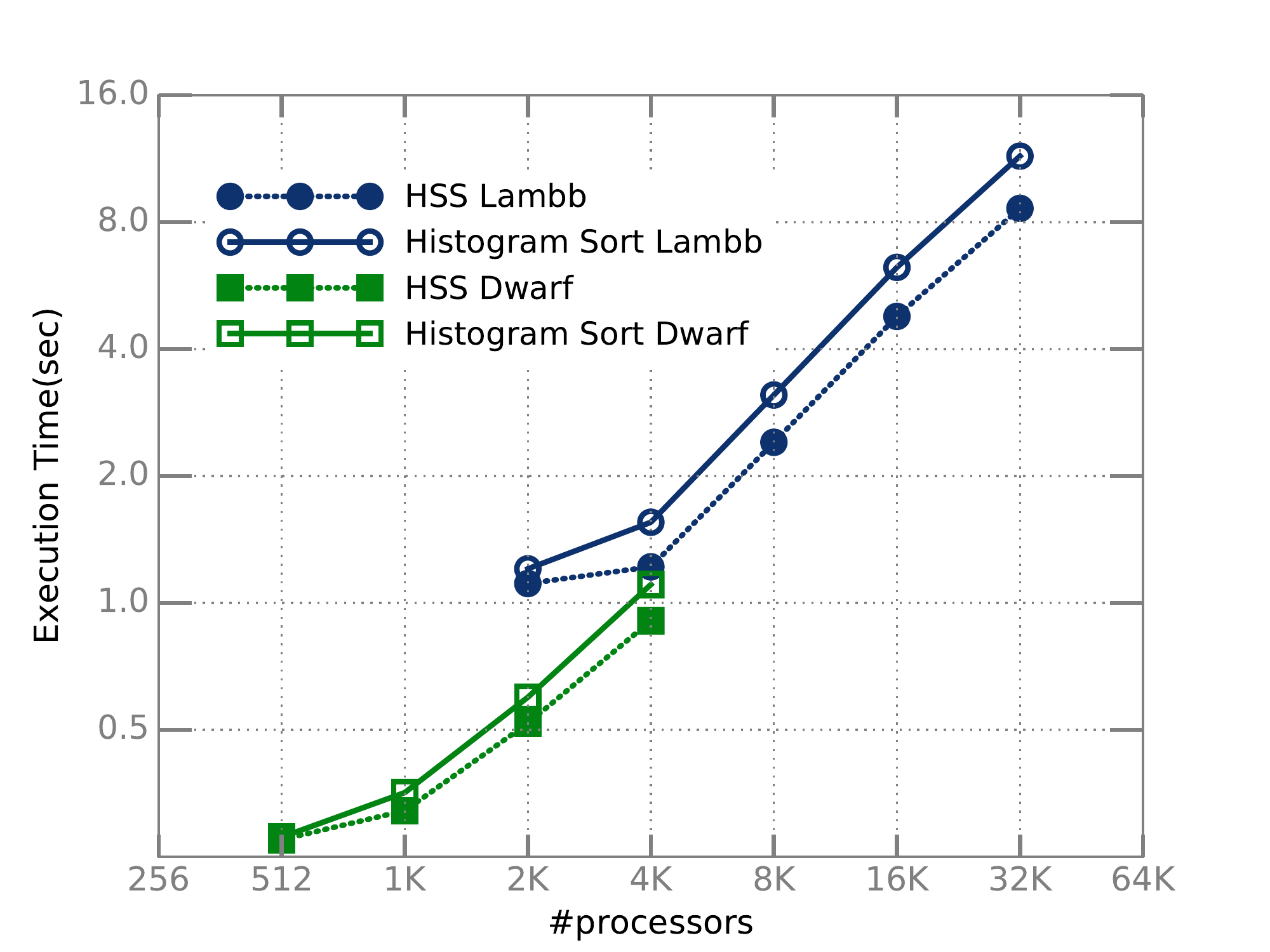}
	\compactcaption{Experiments showing performance of sorting routine of Changa. Datasets used were Lambb and Dwarf.  Experiments were run on the Bluegene Mira supercomputer.} 
	\label{fig:changa-sorting}
  \end{minipage}
\end{figure*}

\begin{figure}[ht]
  \hskip 0.25cm
    \begin{minipage}[b]{\linewidth}
    \centering
      \vskip -0.2cm
	\includegraphics[width=0.78\textwidth]{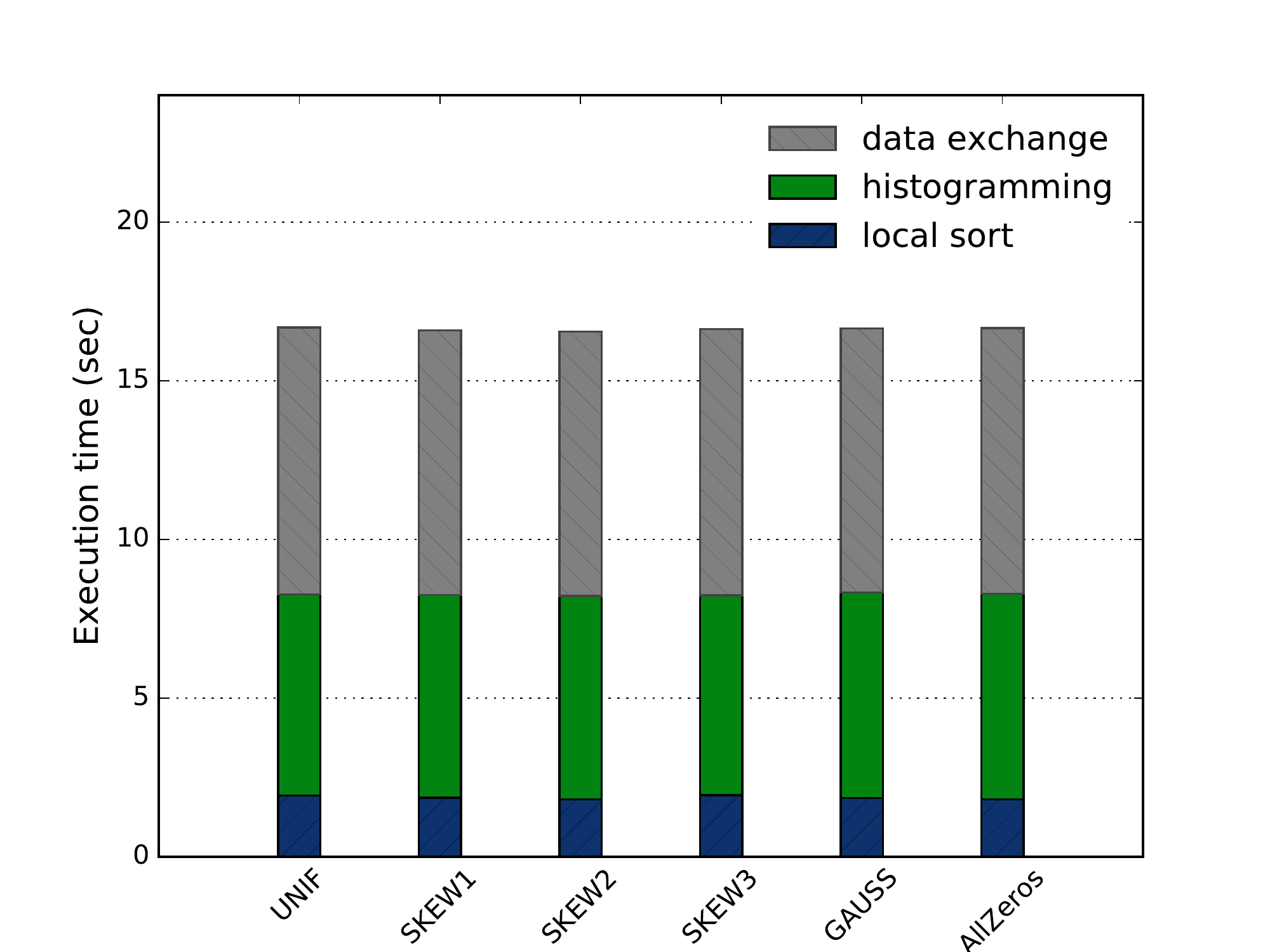}
		\compactcaption{HSS performance with different input distributions for 2M keys/processor and 32K processors. 
		We used 1 thread per processor to accentuate the splitter selection (histogramming) time.  Experiments were run on the Bluegene Mira supercomputer.} 
		\label{fig:diffdist} 
  \end{minipage}
\vskip -0.2cm
\end{figure}

\subsection{Multi-staged experiments}
In this section, we present multi-staged experiments (specifically with 2 stages) where data is first distributed among $k$ processor groups consisting of $p/k$ processors each for $k-way$ sorting. Multi-staged sorting is helpful when the number of message startups ($=\ p$ messages) per processor becomes a bottleneck. This happens for a large number of processors or when the number of keys per processor is small enough that very fine grained messages have to be sent to other processes in the data exchange step which slows down the sorting operation. Note that there is an overhead associated with multi-staged sorting as data needs to be moved multiple times in comparison to single-staged sorting where data is exchanged just once between the source and the destination processors. 

For this set of experiments, we used $10^5$ keys per processor and 1 thread per processor. We used $k=128$ as we found it to be a reasonable threshold for using 2-staged sorting. It also happens to be the recommended setting as per ~\cite{sundar2013hyksort}. Note that $\lceil\sqrt{p}\ \rceil = 128$ for $p=16384$. We used a tolerance threshold $\epsilon=1\%$ per stage so that the overall imbalance is at most $2\%$. A single thread per process was used to maximize the number of processors. Accordingly, we also scaled down the number of keys/process when compared to section~\ref{single stage}. This also kept the number of keys/thread comparable to section~\ref{single stage} (it is slightly higher in this case).

Figure~\ref{fig:multiplestage} illustrates 2-staged runs for $p=8192$ and $p=16384$ processors. As can be seen from the figure, multiple stage sorting alleviates the data exchange bottleneck and the splitter determination step becomes the major bottleneck. Figure~\ref{fig:multiplestage} demonstrates the benefit of using HSS in a multi-staged setting. HSS improves the overall running time by $15-20\%$ for both $p=8192$ and $p=16384$. We verified that this improvement comes from improving the number of iterations for convergence in each stage: HSS converged in 6 iterations while Hyksort took 9 iterations.

\subsection{Strong scaling in ChaNGa}
\ifchangasorting
We implemented HSS in ChaNGa~\cite{jetley2008massively}, a popular astronomical application that often runs on several thousands of processors. Sorting in ChaNGa poses unique challenges for two reasons- (i) it employs virtual processors and hence the number of buckets (virtual processors) are far more than the actual number of processors. In our experiments, the number of virtual buckets were typically \textbf{~10x} the number of cores and (ii) the virtual processors can be arbitrarily placed across physical nodes and buckets on a single node need not be contiguous. Hence, most of our shared memory optimizations are not useful. The reason ChaNGa uses more virtual processors than cores is to accelerate other stages of computation, made possible by efficient parallel data overpartitioning. 

Figure~\ref{fig:changa-sorting} compares sorting performance of ChaNGa with HSS and the existing Histogram sort implementation for two datasets: Dwarf and Lambb (see ~\cite{jetley2008massively} for details). The datasets have a constant number of keys, so Figure~\ref{fig:changa-sorting} represents strong scaling results. HSS results in up to 25\% improvement over Histogram sort. Note that Histogram sort is much more sensitive to the input distribution than HSS as it does not employ sampling (see section~\ref{input distribution}).
The parallel sorting execution increases for the same dataset as we increase the number of processors. This may appear odd at first. The majority of sorting time is spent in data splitting, and since the number of buckets increase multiplicatively with the number of processors, we see an increase in the execution time. 
The performance results suggest it would be possible to improve strong scaling of the splitting algorithm within ChaNGa by using a multi-staged version of HSS. We leave this for future work. 
\fi

\subsection{Effect of input distribution} \label{input distribution}

We ran HSS with the following input distributions to verify that its running time is independent of the distribution:
\begin{enumerate}[leftmargin=*]
\item UNIF: Uniformly at random from the entire range
\item SKEW1: Half of the keys are picked uniformly at random from the entire range, the other half, uniformly at random from a small range of size 1000
\item SKEW2: Uniformly at random from the range $[0,100]$
\item SKEW3: Each key is bitwise and of two uniformly at random chosen keys
\item GAUSS: Gaussian distributed
\item AllZeros: All keys are set to 0
\end{enumerate}
As figure~\ref{fig:diffdist} illustrates, HSS is impervious to the input distribution as expected from the analysis. To underscore the histogramming cost, we used 1 thread per process since the number of splitters for process level splitting is equal to the number of processes.

\section{Conclusion}


We presented Histogram sort with sampling (HSS), which combines sampling and histogramming to accomplish fast splitter determination. 
We showed that for approximate-splitting ($\epsilon=\mathcal{O}(1)$), our algorithm requires $\Theta(\log \log p)$ histogramming rounds  and an overall sample of size $\Theta(p\log \log p)$, improving the communication cost by a factor of $\Theta(\log p/\log \log p)$ with respect to the best known partitioning algorithm. HSS is theoretically more efficient for both approximate and exact (memory-efficient) splitting, while minimizing the number of data exchanges for both small and large degrees of parallelism.
Our work provides theoretical groundwork for the benefits of iterative histogramming in splitter selection, a technique that is known to work well in practice. The reduced sample size makes HSS extremely practical for massively parallel applications, scaling to tens of thousands of processors. 
We demonstrated speed-ups with HSS over two other state-of-the-art parallel sorting implementations for both single-staged and multi-staged settings.
The robustness of our results makes a compelling case for HSS as the algorithm of choice for large scale parallel sorting.

\begin{acks}
  The authors are thankful to Omkar Thakoor, Umang Mathur and Harshita Prabha for helpful discussions and Nitin Bhat for helping with experiments. The authors acknowledge the Argonne Leadership Computing Facility (ALCF) and the Texas Advanced Computing Center (TACC) at The University of Texas at Austin  for providing HPC resources that have contributed to the research results reported within this paper. URL: \url{https://www.alcf.anl.gov/} and \url{http://www.tacc.utexas.edu}
\end{acks}


%% file: appendix.tex

\appendix


\input{scanning_algorithm_analysis}

\input{hyksort_proof_include}


%% file: scanning_algorithm_analysis.tex

\label{apx:A}
\section{Scanning Algorithm analysis}

We formalize the proof sketch given by Axtmann et al.~\cite{axtmann2015practical} for their parallel sorting algorithm.
Their algorithm scans the results of a histogram collected from a randomly selected sample, picking splitters consecutively so that each processor gets no more than $(1+\epsilon)N/p$ elements.
If the sample is large enough, with high probability the average load assigned to each consecutive processor will be greater than $N/p$ and the last processor will not end up with too many keys.
We provide a formal bound on the probability that this holds for the appropriate sample size.
\begin{lemma} \label{one round scanning algorithm}
If every key is independently picked in the sample with probability, ${ps}/{N} = {p} \max({2}/{\epsilon}, 2\log N)/N$, then the scanning algorithm assigns no more than ${N(1+\epsilon)}/{p}$ keys to each processor with high probability.
\end{lemma}
\begin{proof}
Let $\srank{0},\ldots,\srank{{}p}$ be the indices in the global order of the splitter keys selected by the scanning algorithm (with $\srank{0}=0$ and $\srank{p}=N$) and let $n_i=\srank{{}i+1}-\srank{{}i}$ be the resulting number of keys asisgned to processor $i$.
By construction the algorithm assigns at most $n_i\leq (1+\epsilon)N/p$ for $i< p-1$, so it suffices to show $n_{p}\leq (1+\epsilon)N/p$.

Select all unique splitter indices $\{\tilde{\chi}_0,\ldots,\tilde{\chi}_h\}$ from $\{\srank{0},\ldots,\srank{{}p-1}\}$, so $h < p$. Note that we exclude $\srank{{}p}$ and hence $\tilde{\chi}_{i+1} - \tilde{\chi}_{i} < N(1+\epsilon)/p $ for all $i$.
Splitters can repeat if there are more than $(1+\epsilon)N/p$ keys between any pair of splitters (including splitters with index $0$ and $N$), in which case the algorithm essentially fails.
The only other way the scanning algorithm yields repeated splitters is if it runs out of sample keys to assign to pick as splitters, in which case $n_{p-1}=0$.
Without loss of generality, we assume the latter does not happen and show that the former is improbable by showing that w.h.p.\ a sample is selected from each subrange of size $N/p$,
\begin{align*}
P(h=p-1) & \geq 1 - \Big( N \times \big(1-\frac{ps}{N}\big)^{\frac{N}{p}} \Big) \\
& \geq 1 - \Big( N \times e^{-\frac{Nps}{pN}} \Big)  =  1 - \Big( N \times e^{-s} \Big) \\
& \geq 1  - \Big( N \times e^{-2\log N} \Big) \geq 1 - \frac{1}{N}
\end{align*}
We proceed to show that if $h=p-1$, then $n_{p-1}$ is small with high probability.

Let $t_i$ be the next key in the sample after $\tilde{\chi}_i$ for $i\leq p$.
We have that $\tilde{\chi}_{i+1}\geq t_i$ for $i<p$.
Given the location of splitter $\tilde{\chi}_{i-1}$, we have that $\tilde{\chi}_i\leq o_i=\tilde{\chi}_{i-1}+(1+\epsilon)N/p$.
Since $t_i$ and $o_i$ depend only on whether keys in the range $[0,t_i]$ are present in the sample, the probability of any key in the range $(t_i,o_i]$ being in the sample, is independent of $t_i$ and $o_i$ given the size of the interval ($o_i - t_i$).
Define $r_i=o_i-\tilde{\chi}_i=(1+\epsilon)N/p - (\tilde{\chi}_i-\tilde{\chi}_{i-1})$.
Each $r_i\leq o_i-t_i$ is a random variable that depends exclusively on the existence and location of a sample key in $(t_i,o_i]$ and the size of the interval ($o_i-t_i$).
Since the probability of observing keys in this interval is independent of $t_i$ given the size of the interval ($o_i - t_i$), $r_i$ can depend on $\{r_1,\ldots,r_{i-1}\}$ only as a function of the size of the interval. 

We can express each $r_i$ as $\text{min}(o_i - t_i, \tilde{r}_i)$ where $\tilde{r}_i$ is a random variable corresponding to the distance between the closest key in $[0,(1+\epsilon)N/p]$ to $(1+\epsilon)N/p$ when keys in this range are picked with uniform probability $ps/N$.
Then we observe that $r_i\leq \tilde{r}_i$ and the random variables $\tilde{r}_i$ are mutually independent.
%
%
Since $\tilde{r}_i$s are picked independently using binomial trials, they are exponentially distributed. More specifically, $P\big[\tilde{r}_{i} \geq k\big] = (1-ps/N)^{k}$. Thus, we have,
\begin{align*}
 & E\big[r_i\big] \leq  E\big[	\tilde{r}_i\big] \leq \sum_{k=1}^{\infty}\big(1 - \frac{ps}{N}\big)^{k} \leq \frac{N}{ps} 
\end{align*}

Coming back to bounding the number of keys assigned to the last processor by the scanning algorithm $n_{p-1}$, we have,
\begin{align*}
n_{p-1} &= N - \sum_{i=0}^{p-2}n_{i} = N - \sum_{i=0}^{p-2} \big(\frac{N}{p} + \frac{N\epsilon}{p} - r_{i}\big) \\
&= \frac{N}{p} - \frac{N\epsilon(p-1)}{p} + \sum_{i=0}^{p-2} r_{i}
\end{align*}


Using sampling ratio $s = \frac{2}{\epsilon}$ we have,
\begin{align*}
P\Big[\sum_{i=0}^{h-1}r_i \geq N\epsilon\Big] 
&\leq\ P\Big[\sum_{i=0}^{h-1}\tilde{r}_i \geq N\epsilon\Big] \\
&\leq P\Big[\sum_{i=0}^{h-1}\tilde{r}_i \geq \frac{N\epsilon}{2} + \frac{N\epsilon(h-1)}{2p}\Big] \\
&\leq  P\Bigg[\Big(\sum_{i=0}^{h-1}\tilde{r}_i - E[\sum_{i=0}^{h-1}\tilde{r}_i]\Big)\geq \frac{N\epsilon}{2} \Bigg] \\
&\leq e^{\frac{-p\epsilon^{2}}{2(1+\epsilon)^{2}}}
\end{align*}
The last inequality is obtained using an application of Hoeffding's inequality since all $\tilde{r}_i$'s are independent.

With high probability, we have $h=p-1$. Hence using expression for $n_{p-1}$
\begin{align*}
P\Big[n_{p-1} \geq  \frac{N}{p} + \frac{N\epsilon}{p} \Big] 
 &\leq\ P\Big[\sum_{i=0}^{p-2}r_i \geq N\epsilon\Big] \\
&\leq e^{\frac{-p\epsilon^{2}}{2(1+\epsilon)^{2}}} \ \text{w.h.p.}
\end{align*}

\end{proof}


%% file: hyksort_proof_include.tex

\section{HykSort Sampling Algorithm: Analysis} \label{hyksort analysis}

HykSort~\cite{sundar2013hyksort} selects $\mathcal{O}(\beta)$ samples from every splitter interval in every round, thus resulting in an overall sample size $\mathcal{O}(\beta p)$. In contrast, HSS picks samples uniformly from the union of all splitter intervals, also resulting in an overall sample size $\mathcal{O}(\beta p)$. Effectively, sampling in HSS from a splitter interval is proportional to the size of the interval. We prove that HykSort requires at least $\Omega \Big(\frac{\log (p)}{\log\log (p)}\Big)$ rounds, so that all splitters are within a distance of $N\epsilon/p$ from the ideal splitters. 

Our proof strategy is the following. First of all we reduce the problem by using $\beta=1$. Sampling $\beta$ samples per round can bring down the number of rounds by at most a factor of $\beta$. Since we're only interested in the dependence of $p$ on the number of rounds, it suffices to show that HykSort requires at least $\Omega \Big(\frac{\log (p)}{\log\log (p)}\Big)$ rounds with $\beta = 1$.

Secondly, we assume a better starting point for the splitter intervals. More specifically, we assume that the initial $i^{th}$ splitter interval is given to be $[Ni/p - N/2p, Ni/p + N/2p]$, instead of the entire range $[0, N]$. Starting with a narrowed splitter interval will only decrease the number of rounds. This eases the analysis since effectively each splitter interval is being independently sampled and the number of rounds should be enough so that for all $i$, at least one key is sampled that is within the target range $\mathcal{T}_i = [Ni/p - N\epsilon/p, Ni/p + N\epsilon/p]$.

From here on, we can work with just one interval and determine the number of rounds required so that at least one key is chosen in the target range $\mathcal{T}_i = [Ni/p - N\epsilon/p, Ni/p + N\epsilon/p]$ with probability $\geq 1 - 1/p$. Note that probability $\geq 1 - 1/p$ is required to use the union bound to bound the probability of not finding a sample in the target range for any of the splitter intervals (there are $p-1$ splitters to be determined). 

HykSort's sampling algorithm is as follows. In round $r$, it samples one key $k$ (recall that we assumed $\beta = 1$) in the splitter interval $[L_r(i), U_r(i)]$ and then updating the splitter interval as \[U_{r+1}(i) = min(U_r(i), k)\] \[L_{r+1}(i) = max(L_r(i), k)\]

As discussed earlier, the initial interval is \[[L_0(i), U_0(i)] = \Big[\frac{Ni}{p} - \frac{N}{2p}, \frac{Ni}{p} + \frac{N}{2p}\Big]\]

In the following section we prove that it takes $r = \Omega \Big(\frac{\log (p)}{\log\log (p)}\Big)$ rounds such that $P[U_{r}(i) \leq Ni/p + N\epsilon/p]$ with probability $\geq 1 - 1/p$. 

By basically the same argument it can be shown that it takes, $r = \Omega \Big(\frac{\log (p)}{\log\log (p)}\Big)$ rounds such that  $P[L_{r}(i) \geq Ni/p - N\epsilon/p]$ with probability $\geq 1 - 1/p$.

\subsection{The line algorithm}

\begin{figure}
\vskip 1.2 cm
\includegraphics[width=0.35\textwidth, height=37pt]{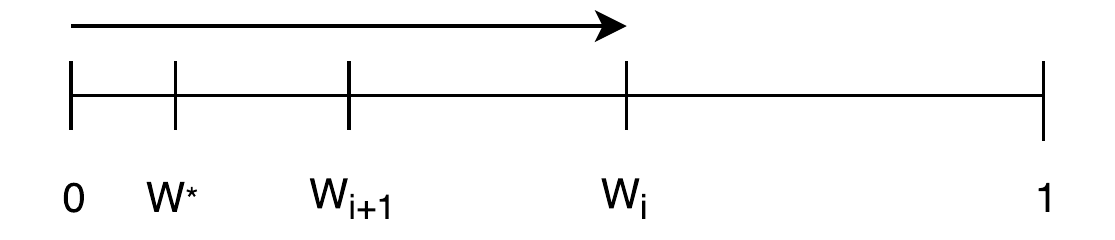}
\compactcaption{Line Algorithm : $w_{i+1}$ is picked uniformly at random from $[0, w_i)$}.
\label{fig:samplesize}
\end{figure}

Consider the following algorithm.

Line Algorithm: Pick a point $w_0$ uniformly at random in the real interval $[0, 1)$. In the next round pick a point $w_1$ uniformly at random in the real interval $[0, w_0)$. Similarily, in the $\text{i}^{\text{th}}$ round pick a random point $w_i$ in the interval $[0, w_{i-1})$ and so on. The line algorithm captures the sampling algorithm of HykSort.

Given a point $w^{*} \in [0,1)$ and probability bound $p^{*}$, we wish to bound the number of rounds $r$ so that $P[w_r > w^{*}] < p^{*}$. For the analysis of HykSort, we'll set $w^{*} = \frac{N\epsilon/p}{N/p} = \epsilon$ and $p^{*} = 1/p$. We prove the following lemma. 

\begin{lemma}
For a given $w^{*} \in [0,1)$, \ the number of rounds $r$ after which $P[w_r > w^{*}] < p^{*}$ is $\Omega \Big(\frac{\log (1/p^{*})}{\log\log (1/p^{*})}\Big)$.
\end{lemma}

\begin{proof}
By definition of the line algorithm, \[0 \leq w_{i+1} \leq w_{i}\]

Let $f^{i}(x)$ be the probability density function of $w_{i}$. We have the following recurrence for $f^{i}(x)$,

\begin{align*}
f^{i}(x) = \int_x^1 \frac{f^{i-1}(y)}{y} dy \quad \forall i \geq 1
\end{align*}

The expression inside the integral represents the probability $P\big[w_{i} \in [x, x+dx]\ \big|\ w_{i-1} \in [y, y+dy]\big] P\big[w_{i-1} \in [y, y+dy]\big]$. We have $f^{0}(x) = 1$. Using induction on $i$, it can be easily seen that \[f^{i}(x) = \frac{\log^{i} (\frac{1}{x})}{i!}\]

We can also obtain the corresponding cumulative density functions $F^{i}(x)$,

\begin{align*}
F^{i}(x) &= P[w_{i} \leq x] \\
&= \int_0^x f^{i}(y) dy \\
&= x \sum_{k=0}^i \frac{\log^{k} \big(\frac{1}{x}\big)}{k!}
\end{align*}

It can be verified that $\lim_{i \to \infty} F^{i}(x) =  x e^{\log (\frac{1}{x})} = 1$, using Taylor's expansion for the exponential function. We have

\begin{align*}
P[w_r > w^{*}] &= 1 - F^r(w^{*}) \\
&= 1 - w^{*} \sum_{k=0}^r \frac{\log^{k} \big(\frac{1}{w^{*}}\big)}{k!} \\
&= \Big(1 - e^{t} \sum_{k=0}^r \frac{t^{k}}{k!}\Big), \text{ where $t = \log \frac{1}{w^{*}}$}\\
&= e^{-t} \Big(e^{t} - \sum_{k=0}^r \frac{t^{k}}{k!}\Big)\\
&= e^{-t} \Big(\frac{e^{\xi}t^{r+1}}{(r+1)!}\Big) \text{ \ for some $\xi \in [0,t]$}
\end{align*}

The last deduction is based on the error term in the Taylor expansion. We want the error term to be smaller than $p^{*}$. Hence we have,

\begin{align*}
& e^{-t} \Big(\frac{e^{\xi}t^{r+1}}{(r+1)!}\Big) \leq p^{*} \\
\Rightarrow \quad & \frac{(r+1)!}{t^{r+1}} \geq \frac{1}{p^{*}e^{t-\xi}} \\
\Rightarrow \quad & \log (r+1)! - (r+1)\log t \geq \log \frac{1}{p^{*}} - (t-\xi) \\
\Rightarrow \quad & (r+1)\log (r+1) - (r+1) + \mathcal{O}(r) \geq \log \frac{1}{p^{*}} - (t-\xi)
\end{align*}

The last deduction is using Stirling's formula: \[\log n! = n\log n - n + \mathcal{O}(\log n)\] Note that the dominating term in LHS is $(r+1)\log (r+1)$ as $t$ is a constant. Thus we obtain \[r = \Omega \Big(\frac{\log (1/p^{*})}{\log\log (1/p^{*})}\Big)\]

\end{proof}